\documentclass[a4paper,11pt]{amsart}
\usepackage{fullpage}

\usepackage{algorithm}
\usepackage[noend]{algpseudocode}

\makeatletter
\def\BState{\State\hskip-\ALG@thistlm}

\usepackage{bm}
\usepackage{amssymb}
\usepackage{pdfpages}
\usepackage{amsmath}
\usepackage{amsthm}
\usepackage{bbm}
\usepackage{tikz}
\usepackage{paralist}
\usepackage[colorinlistoftodos,bordercolor=orange,
            backgroundcolor=orange!20,linecolor=orange,
            textsize=normalsize]{todonotes}
\usepackage{array}
\newcolumntype{L}[1]{>{\raggedright\let\newline\\\arraybackslash\hspace{0pt}}m{#1}}
\newcolumntype{C}[1]{>{\centering\let\newline\\\arraybackslash\hspace{0pt}}m{#1}}
\newcolumntype{R}[1]{>{\raggedleft\let\newline\\\arraybackslash\hspace{0pt}}m{#1}}

\usepackage[capitalize]{cleveref}

\crefformat{equation}{\textup{#2(#1)#3}}
\crefrangeformat{equation}{\textup{#3(#1)#4--#5(#2)#6}}
\crefmultiformat{equation}{\textup{#2(#1)#3}}{ and \textup{#2(#1)#3}}
{, \textup{#2(#1)#3}}{, and \textup{#2(#1)#3}}
\crefrangemultiformat{equation}{\textup{#3(#1)#4--#5(#2)#6}}%
{ and \textup{#3(#1)#4--#5(#2)#6}}{, \textup{#3(#1)#4--#5(#2)#6}}{, and \textup{#3(#1)#4--#5(#2)#6}}

\Crefformat{equation}{Equation~#2\textup{(#1)}#3}
\Crefrangeformat{equation}{Equations~\textup{#3(#1)#4--#5(#2)#6}}
\Crefmultiformat{equation}{Equations~\textup{#2(#1)#3}}{ and \textup{#2(#1)#3}}
{, \textup{#2(#1)#3}}{, and \textup{#2(#1)#3}}
\Crefrangemultiformat{equation}{Equations~\textup{#3(#1)#4--#5(#2)#6}}%
{ and \textup{#3(#1)#4--#5(#2)#6}}{, \textup{#3(#1)#4--#5(#2)#6}}{, and \textup{#3(#1)#4--#5(#2)#6}}

\newtheorem{theorem}{Theorem}
\newtheorem{proposition}[theorem]{Proposition}
\newtheorem{corollary}[theorem]{Corollary}
\newtheorem{lemma}[theorem]{Lemma}
\theoremstyle{definition}
\newtheorem{definition}[theorem]{Definition}

\theoremstyle{remark}
\newtheorem{remark}[theorem]{Remark}
\newtheorem{example}[theorem]{Example}

\makeatletter
\newcommand{\addresseshere}{%
  \enddoc@text\let\enddoc@text\relax
}
\makeatother

\newfont{\pogrubianemat}{msbm10}
\newfont{\malepogrubianemat}{msbm7}

\def\RR{\mathbb R}

\def\IND{\mathbbm{1}}

\DeclareMathAlphabet{\mathbfcal}{OMS}{cmsy}{b}{n}
\makeatletter
\newcommand\xleftrightarrow[2][]{%
  \ext@arrow 9999{\longleftrightarrowfill@}{#1}{#2}}
\newcommand\longleftrightarrowfill@{%
  \arrowfill@\leftarrow\relbar\rightarrow}
\makeatother

\DeclareTextFontCommand{\textbfit}{%
  \fontseries\bfdefault 
  \itshape
}

\newlength{\leftstackrelawd}
\newlength{\leftstackrelbwd}
\def\leftstackrel#1#2{\settowidth{\leftstackrelawd}%
{${{}^{#1}}$}\settowidth{\leftstackrelbwd}{$#2$}%
\addtolength{\leftstackrelawd}{-\leftstackrelbwd}%
\leavevmode\ifthenelse{\lengthtest{\leftstackrelawd>0pt}}%
{\kern-.5\leftstackrelawd}{}\mathrel{\mathop{#2}\limits^{#1}}}

\newcommand{\des}{\mathbf{\sf{D}}}
\newcommand{\adv}{\mathbf{\sf{A}}}
\newcommand{\Varphi}{\mathit{\Phi}}
\newcommand{\Vardelta}{\mathit{\Delta}}
\newcommand{\Vargamma}{\mathit{\Gamma}}

\newcommand{\Ex}{\mathbf{E}}
\newcommand{\poa}{\mathrm{PoA}}

\newcommand{\graph}{G}
\newcommand{\trgraph}{G'}
\newcommand{\defense}{\Delta}
\newcommand{\defcomp}{C_{1}}
\newcommand{\trdefcomp}{C'_{1}}
\newcommand{\comp}{C}

\newcommand{\abs}[1]{|{#1}|}
\newcommand{\Payoff}{P}

\newcommand{\PoA}{\mathrm{PoA}}

\newcommand{\overbar}[1]{\mkern 1.5mu\overline{\mkern-1.5mu#1\mkern-1.5mu}\mkern 1.5mu}

\newcommand{\byznode}{\theta_{\adv}}

\renewcommand{\subset}{\subseteq}

\newcommand{\globdespayoff} { \hat{U}^{\des}_{\star}  } 
\newcommand{\equilibrium}{\bm e}
\newcommand{\Equilibria}{\mathcal{E}}
\newcommand{\graphs}{\mathcal{G}}
\newcommand{\nodes}{V}
\newcommand{\simplex}{\Sigma}
\newcommand{\card}{\abs}

\newcommand{\despayoff}{w}
\newcommand{\byznodes}{B}
\newcommand{\nbyz}{n_B}
\newcommand{\ninf}{n_A}

\title{Individual Security and Network Design with Malicious Nodes}
\date{\today}

\author{Tomasz Janus}
\author{Mateusz Skomra}
\author{Marcin Dziubi{\'n}ski}

\address{\textit{T.~Janus, M.~Dziubi{\'n}ski}: Institute of Informatics, Faculty of Mathematics, Informatics and Mechanics, University of Warsaw, Banacha 2, 02-097 Warszawa, Poland}
\email{t.janus@mimuw.edu.pl, m.dziubinski@mimuw.edu.pl}

\address{\textit{M.~Skomra}: CMAP, Ecole Polytechnique, CNRS and INRIA, 91128 Palaiseau Cedex, France}
\email{mateusz.skomra@polytechnique.edu}


\thanks{This work was supported by Polish National Science Centre through grant no 2014/13/B/ST6/01807. M.~Skomra is  supported by a grant from R{\'e}gion Ile-de-France.}

\begin{document}
\begin{abstract}
Networks are beneficial to those being connected but can also be used as carriers of contagious hostile attacks. These attacks are often facilitated by exploiting corrupt network users. To protect against the attacks, users can resort to costly defense. The decentralized nature of such protection is known to be inefficient but the inefficiencies can be mitigated by a careful network design. Is network design still effective when not all users can be trusted? We propose a model of network design and defense with byzantine nodes to address this question. We study the optimal defended networks in the case of centralized defense and, for the case of decentralized defense, we show that the inefficiencies due to decentralization can be fully mitigated, despite the presence of the byzantine nodes.
\end{abstract}

\maketitle

\section{Introduction}
\label{sec:intro}
Game theoretic models of interdependent security have been used to study security of complex information and physical systems for more than a decade~\cite{LFB14}. One of the key findings is that the externalities resulting from security decisions made by selfish agents lead to, potentially significant, inefficiencies. This motivates research on methods for improving information security, such as insurance~\cite{BS10} and network design~\cite{CDG14,CDG17}. We study the problem of network design for interdependent security in the setup where a strategic adversary collaborates with some nodes in order to disrupt the network.

\subsection*{The motivation}
Our main motivation is computer network security in face of contagious attack by a strategic adversary. Examples of contagious attacks are stealth worms and viruses, that gradually spread over the network, infecting subsequent unprotected nodes. Such attacks are considered among the main threats to cyber security~\cite{SPW02}. Moreover, the study of the data from actual attacks demonstrates that the attackers spend time and resources to study the networks and choose the best place to attack~\cite{SPW02}. Direct and indirect infection can be prevented by taking security measures that are costly and effective (i.e., provide sufficiently high safety to be considered perfect). Examples include using the right equipment (such as dedicated high quality routers), software (antivirus software, firewall), and following safety practices. All of these measures are costly. In particular, having antivirus software is cheap but using it can be considered to be costly, safety practises may require staff training, staying up to date with possible threats, creating backups, updating software, hiring specialized, well-paid staff. The security decisions are made individually by selfish nodes. Each node derives benefits from the nodes it is connected to (directly or indirectly) in the network. An example is the Metcalfe's law (attributed to Robert Metcalfe~\cite{SV00}, a co-inventor of Ethernet) stating that each node's benefits from the network are equal to the number of nodes it can reach in the network, and the value of a connected network is equal to the square of the number of its nodes. An additional threat faced by the nodes in the network is the existence of malicious nodes whose objectives are aligned with those of the adversary: they aim to disrupt the network~\cite{MSW06,MSW09}.

\subsection*{Contribution}
We study the effectiveness of network design for improving system security with malicious (or byzantine) players and strategic adversary.
To this end we propose and study a three stage game played by three classes of players: the designer, the adversary, and the nodes. Some of the nodes are malicious and cooperate with the adversary. The identity of the nodes is their private information, known to them and to the adversary only. The designer moves first, choosing the network of links between nodes. Then, costly protection is assigned to the nodes. We consider two methods of protection assignments: the centralized one, where the designer chooses the nodes to receive protection, and the decentralized one, where each node decides individually and independently whether to protect or not. Lastly, the adversary observes the protected network and chooses a single node to infect. The protection is perfect and each non-byzantine node can be infected only if she is unprotected. The byzantine nodes only pretend to use the protection and can be infected regardless of whether they are protected or not. After the initial node is infected, the infection spreads to all the nodes reachable from the origin of infection via a path containing unprotected or byzantine nodes.
We show that if the protection decisions are centralized, so that the designer chooses both the network and the protection assignment, then either choosing a disconnected network with unprotected components of equal size or a generalized star with protected core is optimal. When protection decisions are decentralized, then, for sufficiently large number of nodes, the designer can resort to choosing the generalized star as well. In the case of sufficiently well-behaved returns from the network (including for example Metcalfe's law), the protection chosen by the nodes in equilibrium guarantees outcomes that are asymptotically close to the optimum. Hence, in such cases, the inefficiencies due to defense decentralization can be fully mitigated even in the presence of byzantine nodes.

\subsection*{Related work}
There are two, overlapping, strands of literature that our work is related to: the interdependent security games~\cite{LFB14} and multidefender security games~\cite{SVL14,LV15,LSV17}.
Early research on interdependent security games assumed that the players only care about their own survival and that there are no benefits from being connected \cite{KH03,V04,ACY06,LB08a,LB08b,CCO12,AMO16}. In particular, the authors of \cite{ACY06} study a setting in which the network is fixed beforehand, nodes only care about their own survival, attack is random, protection is perfect, and contagion is perfect: infection spreads between unprotected nodes with probability $1$. The focus is on computing Nash equilibria of the game and estimating the inefficiencies caused by defense decentralization. They show that finding one Nash equilibrium is doable in polynomial time, but finding the least or most expensive one is NP-hard. They also point out the high inefficiency of decentralized protection, by showing unboundedness of the price of anarchy. In~\cite{LB08a,LB08b} techniques based on local mean field analysis are used to study the problem of incentives and externalities in network security
on random networks. In a more recent publication~\cite{AMO16}, individual investments in protection are considered. The focus is on the strategic structure of the security decisions across individuals and how the network shapes the choices under random versus targeted attacks. The authors show that both under- and overinvestment may be present when protection decisions are decentralized.
A slightly different, but related, models are considered in~\cite{GWA10,GWA11,GMW12,LSB12a,LSB12b}. In these models
the defender chooses a spanning tree of a network, while the attacker chooses a link to remove. The defender and the adversary move simultaneously. The attack is successful if the chosen link belongs to the chosen spanning tree. Polynomial time algorithms for computing optimal attack and defense strategies are provided for several variants of this game. For a comprehensive review of interdependent security games see an excellent survey~\cite{LFB14}.

Multidefender security games are models of security where two or more defenders make security decisions with regard to nodes, connected in a network, and prior to an attack by a strategic adversary. Each of the defenders is responsible for his own subset of nodes and the responsibilities of different defenders are non-overlapping. The underlying network creates interdependencies between the defenders' objectives, which result in externalities, like in the interdependent security games. The distinctive feature of multidefender security models is the adopted solution concept: the average case Stackelberg equilibrium.  The model is two stage. In the first stage the defenders commit to mixed strategies assigning different types of security configurations across the nodes. In the second stage the adversary observes the network and chooses an attack. The research focuses on equilibrium computation and quantification of inefficiencies due to distributed protection decisions.

Papers most related to our work are~\cite{MSW09,CDG14,CDG17,GJKKM16}. The authors of~\cite{MSW06} introduce malicious nodes to the model of~\cite{ACY06}. The key finding in that paper is that the presence of malicious nodes creates a ``fear factor'' that reduces the problem of underprotection due to defense decentralization. Inspired by~\cite{MSW06,MSW09}, we also consider malicious nodes in the context of network defense. We provide a formal model of the game with such nodes as a game with incomplete information. Our contribution, in comparison to~\cite{MSW09}, lies in placing the players in a richer setup, where nodes care about their connectivity as well as their survival, and where both underprotection (i.e., insufficiently many nodes protect as compared to an optimum) and overprotection (excessively many nodes protect as compared to an optimum) problems are present. This leads to a much more complicated incentives structure. In particular, the presence of malicious nodes may lead to underprotection, as nodes may be unable to secure sufficient returns from choosing protection on their own.

Works~\cite{CDG14,CDG17} consider the problem of network design and defense prior to the attack by a strategic adversary. In a setting where the nodes care about both their connectivity and their survival, the authors study the inefficiencies caused by defense decentralization and how they can be mitigated by network design. The authors show that both underprotection as well as overprotection may appear, depending on the costs of protection and network topology. Both inefficiencies can be mitigated by network design. In particular, the underprotection problem can be fully mitigated by designing a network that creates a cascade of incentives to protect. Our work builds on~\cite{CDG14,CDG17} by introducing malicious nodes to the model. We show how the designer can address the problem of uncertainty about the types of nodes and, at the same time, mitigate the inefficiencies due to defense decentralization. 
Lastly, in~\cite{GJKKM16}, a model of decentralized network formation and defense prior to the attack by adversaries of different profiles is considered. The authors show, in particular, that despite the decentralized protocol of network formation, the inefficiencies caused by defense decentralization are relatively low.

The rest of the paper is structured as follows. In \cref{sec:model} we define the model of the game, which we then analyze in \cref{sec:analysis}. In \cref{sec:extension} we discuss possible modifications of our model. We provide concluding remarks in \cref{sec:concl}. \Cref{ap:centralized} contains the proofs of the most technical results.

\section{The model}
\label{sec:model}
There are $(n+2)$ players: the designer ($\des$), the nodes ($V$), and the adversary~($\adv$). In addition, each of the nodes is of one of two types: a genuine node (type~$1$) or a byzantine node (type~$0$). We assume that there are at least $n = 3$ nodes and that there is a fixed amount $\nbyz \ge 1$ of byzantine nodes. The byzantine nodes cooperate with the adversary and their identity is known to $\adv$. All the nodes know their own type only. On the other hand, the adversary has complete information about the game. We suppose that he infects a subset of $\ninf \ge 1$ nodes. A \emph{network} over a set of nodes $V$ is a pair $G = ( V,E )$, where $E \subseteq \{ij \colon i,j \in V\}$ is the set of undirected links of $G$. Given a set of nodes $V$, $\mathcal{G}(V)$ denotes the set of all networks over $V$ and $\mathcal{G} = \bigcup_{U
\subseteq V} \mathcal{G}(U)$ is the set of all networks that can be formed over $V$ or any of its subsets. The game proceeds in four rounds (the numbers $n \ge 3, \nbyz \ge 1, \ninf \ge 1$ are fixed before the game):
\begin{enumerate}
\item The types of the nodes are realized.
\item $\des$ chooses a network $G \in \mathcal{G}(V)$, where $\mathcal{G}(V)$ is the set of all undirected networks over~$V$.
\item Nodes from $V$ observe $G$ and choose, simultaneously and independently, whether to protect (what we denote by $1$) or not (denoted by $0$). This determines the set of protected nodes
      $\Vardelta$. The protection of the byzantine nodes is fake and, when attacked, such node gets infected and transmits the infection to all her neighbors.
\item $\adv$ observes the protected network $(G,\Vardelta)$ and chooses a subset $I \subset V$ consisting of $\card{I} = \ninf \ge 1$ nodes to infect. The infection spreads and eliminates all unprotected or byzantine nodes reachable from $I$ in $G$ via a path that does not contain a genuine protected node from $\Vardelta$.
This leads to the residual network obtained from $G$ by removing all the infected nodes.
\end{enumerate}

Payoffs to the players are based on the residual network and costs of defense.
The returns from a network are measured by a \emph{network value function}
$\Varphi \colon \bigcup_{U \subseteq V} \mathcal{G}(U) \rightarrow \mathbb{R}$
that assigns a numerical value to each network that can be formed over a subset $U$ of nodes from $V$.

A \emph{path in $G$ between nodes} $i,j \in V$ is a sequence of nodes $i_0,\ldots,i_m \in V$
such that $i = i_0$, $j = i_m$, $m \geq 1$, and $i_{k-1}i_k \in E$ for all $k = 1,\ldots,m$. Node $j$ is \emph{reachable} from node $i$ in $G$ if $i = j$ or there is a path between them in $G$. A \emph{component} of a network $G$ is a maximal set of nodes $C\subseteq V$ such that for all $i,j \in C$, $i \neq j$, $i$ and $j$ are reachable in $G$. The set of components of $G$ is denoted by $\mathcal{C}(G)$.
Given a network $G$ and a node $i \in V$, $C_i(G)$ denotes the component $C \in \mathcal{C}(G)$ such that $i \in C$. Network $G$ is \emph{connected} if $|\mathcal{C}(G)| = 1$.

We consider the following family of network value functions:
\begin{displaymath}
\Varphi(G) = \sum_{C\in \mathcal{C}(G)} f(|C|) \, ,
\end{displaymath}
where the function $f \colon \mathbb{R}_{\geq 0} \rightarrow \mathbb{R}$ is increasing, strictly convex, satisfies $f(0) = 0$, and, for all $x \ge 1$, verifies the inequalities 
\begin{equation}\label{eq:move_half}
\begin{aligned}
f(3x) &\geq 2f(2x) \, , \\ 
f(3x + 2) &\ge f(2x+2) + f(2x+1) \, .
\end{aligned}
\end{equation} 
In other words, the value of a connected network is an increasing and strictly convex function of its size. The value of a disconnected network is equal to the sum of values of its components. These assumptions reflect the idea that each node derives additional utility from every node she can reach in the network. In the last property we assume that these returns are sufficiently large: the returns from increasing the size of a component by $50\%$ are higher than the returns from adding an additional, separate, component of the same size to the network.
Such form of network value function is in line with Metcalfe's law, where the value of a connected network over $x$ nodes is given by $f(x) = x^2$, as well as with Reed's law, where the value of a connected network is of exponential order with respect to the number of nodes (e.g., $f(x) = 2^x-1$).

Before defining payoff to a node from a given network, defense, and attack, we formally define the residual network. Given a network $G = (V,E)$ and a set of nodes $Z \subseteq V$, let $G-Z$ denote the network obtained from $G$ by removing the nodes from $Z$ and their connections from $G$. Thus $G-Z = (V\setminus Z, E[V\setminus Z])$, where $E[V\setminus Z] = \{ij \in E \colon i,j \in V\setminus Z\}$.
Given defense $\Vardelta$ and the set of byzantine nodes $\byznodes$, the graph $A(G\mid \Vardelta,\byznodes) = G - \Vardelta\setminus \byznodes$ is called the \emph{attack graph}. By infecting a node $i \in V$, the adversary eliminates the component of $i$ in the attack graph, $C_i(A(G\mid \Vardelta,\byznodes))$.\footnote{We define $C_i(A(G\mid \Vardelta,\byznodes)) = \varnothing$ for every $i \in \Vardelta\setminus \byznodes$.} Hence, if the adversary infects a subset $I \subset \nodes$ of nodes, then the \emph{residual network} (i.e., the network that remains) after such an attack is $R(G \mid \Vardelta, \byznodes, I) = G - \bigcup_{i \in I}C_i(A(G \mid \Vardelta,\byznodes))$. 

Nodes' information about whether they are genuine or byzantine is private. Similarly, the adversary's information about the identity of the byzantine nodes is private. As usual in games with incomplete information, private information of the players is represented by their \emph{types}. The type of a node $i \in V$ is represented by $\theta_i \in \{0,1\}$ ($\theta_i = 1$ means that $i$ is genuine and $\theta_i = 0$ means that $i$ is byzantine) and the type of the adversary is represented by $\theta_{\adv} \in \binom{\nodes}{\nbyz}$. (If $X$ is a finite set, then we denote by $\binom{X}{t}$ the set of subsets of $X$ of cardinality $t$.) A vector $\bm{\theta} = (\theta_1,\ldots,\theta_n,\theta_{\adv})$ of players' types is called a \emph{type profile}. The type profiles must be consistent so that the byzantine nodes are really known to the adversary. The set of consistent type profiles is $\Theta = \{(\theta_1,\ldots,\theta_n,\theta_{\adv}) \colon \theta_{\adv} = \{i \in V \colon \theta_{i} =  0\} , \card{\theta_{\adv}} = \nbyz \}$.
\begin{remark}
We point out that $\byznodes \subset V$ is the set of byzantine nodes (i.e., the true state of the world) while $\theta_{\adv}$ denotes the beliefs of the adversary. The consistency assumption implies that the beliefs of the adversary are correct and $\theta_{\adv} = \byznodes$.
\end{remark}

The adversary aims to minimize the gross welfare (i.e., the sum of nodes' gross payoffs), which is equal to the value of the residual network. Given a network $G$, the set of protected nodes $\Vardelta$, and the type profile $\bm{\theta} \in \Theta$, the payoff to the adversary from infecting the set of nodes~$I$ is
\begin{align*}
u^{\adv}(G,\Vardelta,I \mid \bm{\theta}) = -\Varphi(R(G \mid \Vardelta, \byznodes, I)) = -\sum_{C \in {\mathcal C}(R(G \mid \Vardelta, \byznodes, I))} f(|C|) \, .
\end{align*}

The designer aims to maximize the value of the residual network minus the cost of defense. Notice that this cost includes the cost of defense of the byzantine nodes. Formally, the designer's payoff from network $G$ under defense $\Vardelta$, the set of infected nodes $I$, and the type profile $\bm{\theta}$ is equal to
\begin{align*}
u^{\des}(G,\Vardelta,I \mid \bm{\theta}) =  \Varphi(R(G\mid \Vardelta,I,\byznodes))  - |\Vardelta| c
 = \left(\sum_{C \in \mathcal{C}(R(G \mid \Vardelta,I,\byznodes))} f(|C|)\right) - |\Vardelta| c \, .
\end{align*}

The \emph{gross payoff} to a genuine (i.e., not a byzantine) node $j \in V$ in a network $G$ is equal to $f(|C_j(G)|)/|C_j(G)|$. In other words, each genuine node gets the equal share of the value of her component. The net payoff of a node is equal to the gross payoff minus the cost of protection. A genuine node gets payoff $0$ when removed. Defense has cost $c\in \mathbb{R}_{>0}$.
The byzantine nodes have the same objectives as the adversary and their payoff is the same as that of $\adv$.
Formally, a payoff to the node $j \in V$ given a network $G$ with defended nodes $\Vardelta$, the set of infected nodes~$I$, and the type profile $\bm{\theta}\in \Theta$ is equal to
 \begin{align*}
u^j&(G,\Vardelta,I \mid \bm{\theta})=\begin{cases}
				  u^{\adv}(G,\Vardelta,I \mid \bm{\theta}), & \textrm{if $\theta_j = 0$}, \\
                  \frac{f(|C_j(R(G\mid \Vardelta,\byznodes,I))|)}{|C_j(R(G\mid \Vardelta,\byznodes,I))|}, & \textrm{if $\theta_j = 1$, $j \notin \Vardelta$}, \\
                  & \textrm{and $j \notin \bigcup_{i \in I}C_i(A(G\mid \Vardelta,\byznodes))$},\\
                  \frac{f(|C_j(R(G \mid \Vardelta,\byznodes,I))|)}{|C_j(R(G\mid \Vardelta,\byznodes,I))|} - c, & \textrm{if $\theta_j = 1$ and $j \in \Vardelta$}, \\
                  0, & \textrm{if $\theta_j = 1$, $j \notin \Vardelta$,} \\
                     & \textrm{and $j \in \bigcup_{i \in I} C_i(A(G\mid \Vardelta, \byznodes))$} \, .
                  \end{cases}
\end{align*}

The adversary and the byzantine nodes make choices that maximize their utility. The designer and the nodes have incomplete information about the game and we assume that they are pessimistic, making choices that maximize the worst possible type realization (cf.~\cite{AB06}). Formally, the \emph{pessimistic utility} of a genuine (i.e., of type $\theta_j = 1$) node $j$ from network $G$, the set of protected nodes $\Delta$, and the set of infected nodes $I$, is
\begin{equation*}
\hat{U}^j(G,\Vardelta,I) = \inf_{(\bm{\theta}_{-j},1)\in \Theta} u^j(G,\Vardelta,I \mid (\bm{\theta}_{-j},1)) \, .
\end{equation*}
Similarly, the pessimistic utility of the designer from network $G$, the set of protected nodes $\Delta$, and the set of infected nodes $I$, is
\begin{equation*}
\hat{U}^{\des}(G,\Vardelta,I) = \inf_{\bm{\theta}\in \Theta} u^{\des}(G,\Vardelta,I \mid \bm{\theta}) \, .
\end{equation*}

To summarize, the set of players is $P = V \cup \{\des,\adv\}$. The set of strategies of player $\des$ is $S^{\des} = \mathcal{G}(V)$.
A strategy of each node $j$ is a function $\delta_j \colon \mathcal{G}(V) \times \{0,1\} \rightarrow \{0,1\}$ that,
given a network $G \in \mathcal{G}(V)$ and a node's type $\theta_j\in \{0,1\}$, provides the defense decision $\delta_j(G, \theta_j)$ of the node. The individual strategies of the nodes determine a function $\Vardelta \colon \mathcal{G}(V)\times \{0,1\}^{V} \rightarrow 2^V$ providing, given a network $G \in \mathcal{G}(V)$ and nodes' types profile $\bm{\theta}_{-\adv}\in \{0,1\}^V$, the set of defended nodes $\Vardelta(G \mid \bm{\theta}_{-\adv}) = \{j\in V \colon \delta_j(G,\theta_j) = 1\}$.
The set of strategies of each node $j \in V$ is $S^j = 2^{\mathcal{G}(V)\times \{0,1\}}$. 
A strategy of player $\adv$ is a function $x \colon \mathcal{G}(V) \times 2^V \times \binom{\nodes}{\nbyz} \rightarrow \binom{\nodes}{\ninf}$ that, given a network $G \in \mathcal{G}(V)$, the set of protected nodes $\Vardelta \subseteq V$, and adversary's type $\theta_{\adv}\in \binom{\nodes}{\nbyz}$, provides the set of nodes to infect $x(G,\Vardelta,\theta_{\adv})$. The set of strategies of player $\adv$ is $S^{\adv} = \binom{\nodes}{\ninf}^{\mathcal{G}(V) \times 2^V \times \binom{\nodes}{\nbyz}}$.

Abusing the notation slightly, we use the same notation for utilities of the players from the strategy profiles in the game. Thus, given a strategy profile $(G,\Delta,x)$ and a type profile $\bm{\theta}$, the payoff to player $j \in V \cup \{\des,\adv\}$ is $u^{j}(G,\Vardelta,x \mid \bm{\theta}) = u^{j}(G,\Vardelta(G),x(G,\Vardelta(G)) \mid \bm{\theta})$, the pessimistic payoff to player $j \in V \setminus \byznodes$ is
\begin{equation}\label{eq:pess_payoff_node}
\hat{U}^j(G,\Vardelta,x) = \inf_{(\bm{\theta}_{-j},1)\in \Theta} u^j(G,\Vardelta(G),x(G, \Vardelta(G)) \mid (\bm{\theta}_{-j},1)) \, ,
\end{equation}
and the pessimistic payoff to the designer is given by 
\begin{equation}\label{eq:pess_payoff_des}
\hat{U}^{\des}(G,\Vardelta,x) = \inf_{\bm{\theta}\in \Theta} u^{\des}(G,\Vardelta(G),x(G, \Vardelta(G)) \mid \bm{\theta}) \, .
\end{equation}

By convention, we say that the pessimistic payoff of the byzantine node is the same as her payoff. We are interested in subgame perfect mixed strategy equilibria of the game with the preferences of the players defined by the pessimistic payoffs. We call them the equilibria, for short. We make the usual assumption that when evaluating a mixed strategy profile, the players 
consider an expected value of their payoffs from the pure strategies. In the case of the designer and the genuine nodes, these are expected pessimistic payoffs.

Throughout the paper we will also refer to the subgames ensuing after a network $G$ is chosen. We will denote such subgames by $\Vargamma(G)$ and call the \emph{network subgames}. We will abuse the notation by using the same letters to denote the strategies in $\Vargamma(G)$ and in $\Vargamma$. The set of strategies of each node $i\in V$ in game $\Vargamma(G)$ is $\{0,1\}^{\{0,1\}}$. 
Given the type profile $\bm{\theta}_{-\adv}\in \{0,1\}^V$, the individual strategies of the nodes determine a function $\Vardelta \colon \{0,1\}^{V} \rightarrow 2^V$ that provides the set of defended nodes $\Vardelta(\bm{\theta}_{-\adv}) = \{j\in V \colon \delta_j(\theta_j) = 1\}$. The set of strategies of the adversary in $\Vargamma(G)$ is $\binom{\nodes}{\ninf}^{2^V\times \binom{\nodes}{\nbyz}}$.

All the key notations are summarized in \cref{tab:notation}.

\begin{table}
\small
\begin{tabular}{|c|c|}
\hline
$n$&\parbox[][20pt][c]{4cm}{\centering number of nodes} \\
\hline 
$\nbyz$&\parbox[][30pt][c]{4cm}{\centering number of byzantine nodes} \\
\hline
$\ninf$&\parbox[][35pt][c]{4cm}{\centering number of nodes infected by the adversary} \\
\hline
$f$&\parbox[][20pt][c]{4cm}{\centering component value function} \\
\hline
$\graph$&\parbox[][20pt][c]{4cm}{\centering network} \\
\hline
$\defense$&\parbox[][20pt][c]{4cm}{\centering set of protected nodes} \\
\hline
$u^{\des}, u^{\adv}, u^{j}$&\parbox[][30pt][c]{4cm}{\centering payoff to the designer, the adversary, and a node} \\
\hline
$\hat{U}^{\des}, \hat{U}^{j}$&\parbox[][30pt][c]{4cm}{\centering pessimistic payoff to the designer and a node} \\
\hline
\end{tabular}
\vspace*{0.5cm}
\caption{Summary of the notation.}\label{tab:notation}
\end{table}

\subsection{Remarks on the model}
We make a number of assumptions that, although common for interdependent security games, are worth commenting on. Firstly, we assume that protection is perfect. This assumption is reasonable when available means of protection are considered sufficiently reliable and, in particular, deter the adversary towards the unprotected nodes. Arguably, this is the case for the protection means used in cybersecurity. Secondly, we assume that the designer and genuine nodes are pessimistic and maximize their worst-case payoff. Such an approach is common in computer science and is in line with trying to provide the worst-case guarantees on system performance. One can also take the probabilistic approach (by supposing that the distribution of the byzantine nodes is given by a random variable). In \cref{sec:extension} we discuss how our results carry over to such model.

\section{The analysis}
\label{sec:analysis} 

We start the analysis by characterizing the centralized defense model, where the designer chooses both the network and the defense assignment to the nodes. After that the adversary observes the protected network and nodes' types and chooses the nodes to infect. We focus on the first nontrivial case $\nbyz = \ninf = 1$. In this case, we are able to characterize networks that are optimal to the designer. The topology of these networks is based on the generalized $k$-stars. We then turn to the decentralized defense and study the cost of decentralization. It turns out that the topology of $k$-star gives asymptotically low cost of decentralization not only for the simple case studied earlier but for all possible values of parameters $\nbyz$ and $\ninf$. This is enough to prove our main result, \cref{th:poa}, providing bounds on the price of anarchy.

\subsection{Centralized defense}\label{sec:centralized}
Fix the parameters $\nbyz, \ninf$ and suppose that the designer chooses both the network and the protection assignment. This leads to a two stage game where, in the first round, the designer chooses a protected network $(G,\Delta)$ and in the second round the adversary observes the protected network and nodes' types (recognizing the byzantine nodes) and chooses the nodes to attack. Payoffs to the designer and to the adversary are as described in \cref{sec:model} and we are interested in subgame perfect mixed strategy equilibria of the game with pessimistic preferences of the designer.
We call them equilibria, for short. Notice that, since the decisions are made sequentially, there is always a pure strategy equilibrium of this game. In this section, we focus only on such equilibria. Furthermore, the equilibrium payoff to the designer is the same for all equilibria. We denote this payoff by $\globdespayoff(n,c)$.

In the rest of this subsection we focus on the case $\nbyz = \ninf = 1$. In this case, when the protection is chosen by the designer, two types of protected networks can be chosen in an equilibrium (depending on the value function and the cost of defense): a disconnected network with no defense or a generalized star with protected core and, possibly, one or two unprotected components. Before stating the result characterizing equilibrium defended networks and equilibrium payoffs to the designer, we need to define the key concept of a generalized star and some auxiliary quantities. We start with the definition of a generalized star. If $G = (V, E)$ is a network and $V' \subset V$ is a subset of nodes, then we denote by $G[V']$ the subnetwork of $G$ induced by $V'$, i.e., the network $G[V'] = (V', \{ij \in E \colon i,j \in V' \})$.

\begin{definition}[Generalized $k$-star]
Given a set of nodes $V$ and $k \ge 1$, a \emph{generalized $k$-star} over $V$ is a network $G = (V,E)$ such that the set of nodes $V$ can be partitioned into two sets, $C$ (the core) of size $|C| = k$ and $P$ (the periphery), in such a way that $G[C]$ is a clique, every node in $P$ is connected to exactly one node in $C$, and every node in $C$ is connected to $\lfloor n/k \rfloor - 1$ or $\lceil n/k \rceil - 1$ nodes in $P$.
\end{definition}

\begin{figure}[t]
\begin{center}
\begin{minipage}{0.45\textwidth}
      \centering
        \includegraphics[scale=0.5]{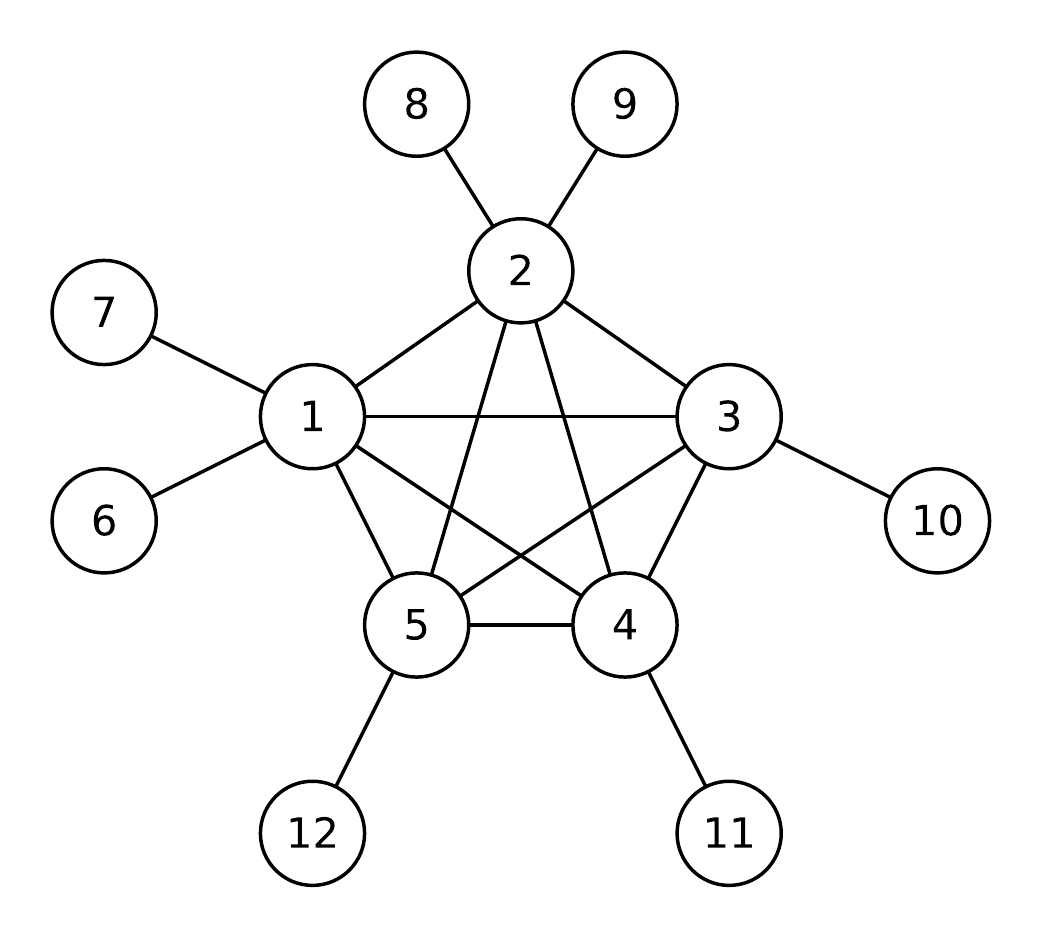}
     \end{minipage}\hfill 
\end{center}
\caption{A generalized star with $12$ nodes and core of size $5$.}\label{fig:star}
\end{figure}

Roughly speaking, a generalized $k$-star is a core-periphery network with the core consisting of $k$ nodes and the periphery consisting of the remaining $n - k$ nodes. The core is a clique, each periphery node is connected to exactly one core node and they are distributed evenly across the core nodes. An example of a generalized star is depicted in \cref{fig:star}.

Now we turn to defining some auxiliary quantities. For any $n \ge 3$ such that $n \bmod 6 \neq 3$ we define
\[
\despayoff_{0}(n) = \despayoff_{1}(n) = f\left(\left\lfloor \frac{n}{2} \right\rfloor\right) + f(1)\IND_{\{n \bmod 2 = 1 \}} \, ,
\]
and for every $n$ such that $n \bmod 6 = 3$ we define
\begin{equation}
\despayoff_{0}(n) = \despayoff_{1}(n) = \max\left( 2f\left( \frac{n}{3} \right), f\left(\frac{n-1}{2} \right) + f(1)\right) \, . \label{eq:divisible_by_three}
\end{equation}
Given $n$ nodes, $\despayoff_{0}(n)$ is the maximal network value the designer can secure against a strategic adversary by choosing an unprotected network composed of three components of equal size or two components of equal size and possibly one disconnected node. This is also the maximal network value the designer can secure by choosing such a network with one protected node, because, in the worst case scenario, the protected node is byzantine and may be infected.

For every $k \in \{3, \dots, n\}$, let
\begin{equation*}
\despayoff_k(n) = \left\{\begin{array}{ll}
                f\left(n-1-\frac{n-1}{k}\right) + f(1), & \textrm{if $n \bmod k = 1$} \, , \\
                f\left(n - \lceil \frac{n}{k} \rceil \right), & \textrm{otherwise} \, .
                \end{array}\right.
\end{equation*}
Given $n$ nodes and $k \ge 3$, $\despayoff_k(n)$ is the network value that the designer can secure by choosing a generalized $k$-star, with one node disconnected in the case of $k$ dividing $n-1$, having all core nodes protected and all periphery nodes unprotected.

We also define the following quantities:
\begin{align}
&A_{q} = \min \left( f(n - q), f\left(\left\lfloor\frac{n - q}{2}\right\rfloor\right) + f(q) \right) \, , \nonumber \\
&B_{q} = \min \left( f(n - q - 1), f\left(\left\lfloor\frac{n-q-1}{2}\right\rfloor\right) + f(q) \right) \, , \nonumber \\
&h_q(n) = \max\left(A_{q}, B_{q} + f(1) \right) \, , \label{eq:size_two_single} \\
&\despayoff_2(n) = \max_{q \in \{0, \ldots, n - 2 \}} h_{q}(n) \, . \label{eq:size_two}
\end{align}
Given $n$ nodes, $\despayoff_2(n)$ is the network value that the designer can secure by choosing a network composed of a generalized $2$-star with a protected core and unprotected periphery, an unprotected component (of size $q \in \{0, \dots, n - 2\}$), and possibly one node disconnected from both of these components.

Finally, we define
\begin{equation*}
K^{*}(n,c) = {\arg\max}_{k\in \{0,\ldots,n\}} \despayoff_k(n) - kc \, .
\end{equation*}

We point out that $K^{*}(n,c)$ never contains $1$ (because $c > 0$). We are now ready to state the result characterizing equilibrium defended network and pessimistic equilibrium payoffs to the designer.

\begin{proposition}
\label{th:centr}
Let $\nbyz = \ninf = 1$, $n \ge 3$,
$c > 0$, and $k \in K^{*}(n,c)$. Then, the pessimistic equilibrium payoff to the designer is equal to $\globdespayoff(n,c) = \despayoff_{k} - kc$. Moreover, there exists an equilibrium network $(G, \Delta)$ that has $\abs{\defense} = k$ protected nodes and the following structure:
\begin{compactenum}[i)]
\item $G$ has at most three connected components.
\item If $k \ge 3$ and $n \bmod k \neq 1$, then $G$ is a generalized $k$-star with protected core and unprotected periphery. 
\item If $k \ge 3$ and $n \bmod k = 1$, then $G$ is composed of a generalized $k$-star of size $(n-1)$ with protected core and unprotected periphery and a single unprotected node.
\item If $k = 0$ and $n \bmod 6 \neq 3$, then $G$ has two connected components of size $\lfloor n/2 \rfloor$ and, if $n \bmod 2 = 1$, a single unprotected node.\label{it:no_def}
\item If $k = 0$ and $n \bmod 6 = 3$, then $G$ either has the structure described in \cref{it:no_def} or $G$ is composed of three components of size $n/3$, depending on the term achieving maximum in~\cref{eq:divisible_by_three}.
\item If $k = 2$, then $G$ is composed of a generalized $2$-star with protected core and unprotected periphery, an unprotected component of size $q \in \{0, \dots, n-2\}$ and, possibly, a single unprotected node. The size $q$ is the number achieving maximum in \cref{eq:size_two}. The existence of a single unprotected node depends on the term achieving maximum in \cref{eq:size_two_single}.
\end{compactenum}
\end{proposition}

The intuitions behind this result are as follows. When the cost of defense is high, then the designer is better off by not using any defense and partitioning the network into several components. Since the strategic adversary will always eliminate a maximal such component, the designer has to make sure that all the components are equally large. Due to the divisibility problems, one component may be of lower size. Thanks to our assumptions on the component value function $f$, the number of such components is at most three. Moreover, if there are exactly three components, then they are of equal size or the smallest one has size $1$.

When the cost of defense is sufficiently low, then it is profitable for the designer to protect some nodes. If the number of protected nodes is not smaller than $3$, then, by choosing a generalized $k$-star with fully protected core (of optimal size $k \ge 3$ depending on the cost) and unprotected periphery, the designer knows that the strategic adversary is going to attack either the byzantine node (if she is among the core nodes) or any unprotected node (otherwise). An attack on the byzantine core node destroys that node and all periphery nodes attached to her. Thus, in the worst case, a core node with the largest number of periphery nodes connected to her is byzantine. By distributing the core nodes evenly, the designer minimizes the impact of this worst case scenario. Due to the divisibility problems, it may happen that some of the core nodes are connected to a higher number of periphery nodes. If this is the case for one core node only, then it is better for the designer to disconnect this one node from the generalized star. By doing so, the designer spares this node from destruction.

The case when there are exactly $2$ protected nodes is special. Indeed, in this case, choosing a generalized $2$-star with protected core is not better than using no protection at all. This is because, in the worst case, the byzantine node is among the two protected ones. Therefore, it would be better for the designer to split the network into two unprotected components -- this would result in the same network value after the attack without the need to pay the cost of protection. On the other hand, if the network consists of a generalized $2$-star with protected core and an unprotected component, then the argument above ceases to be valid: even if the byzantine node is among the protected ones, splitting them may give the adversary an incentive to destroy the unprotected component. Therefore, a protection of $2$ nodes may be used as a resource that ensures that one component survives the attack.

It is interesting to compare this result to an analogous result obtained in~\cite{CDG14,CDG17} for a model without byzantine nodes. There, depending on the cost of protection, three equilibrium protected networks are possible: an unprotected disconnected network (like in the case with a byzantine node), a centrally protected star, and a fully protected connected network. The existence of a byzantine node leads to a range of core-protected networks between the centrally protected star and the fully protected clique (which is a generalized $n$-star). Notice that pessimistic attitude towards incomplete information results in the star network never being optimal: if only one node is protected, then, in the worst case, the designer expects this node to be byzantine, which leads to loosing all nodes after the attack by the adversary. Therefore, at least two nodes must be protected if protection is used in an equilibrium. The proof of \cref{th:centr} is given in \cref{ap:centralized}.

\begin{example}
\Cref{example50} presents how the optimal network changes for different cost values when $f(x) = x^{2}$ and $n \in \{12,30,50\}$. For these values of $n$, it is never optimal to have one node that is disconnected from the rest of the network. Moreover, as we can see, for a given number $n$ of nodes, not all possible generalized $k$-stars arise as optima. It is interesting to note that $3$-stars have never appeared in our experiments as optimal networks for the value function $f(x) = x^{2}$. Similarly, we have not found an example where it is optimal to defend exactly $2$ nodes. The case where there is no defense but the network is split into $3$ equal parts arises when $n = 9$ and the cost is high enough (i.e., $c > 6.2$), as already established in~\cite{CDG14}.
\end{example}

\begin{remark}
In this section, we have characterized the optimal networks for the case $\nbyz = \ninf = 1$. Nevertheless, we have not found a network that has a substantially different structure than the ones described here and performs better for general values of $\nbyz$ and $\ninf$. We therefore suspect that the characterization for the general case is similar to the case $\nbyz = \ninf = 1$.
\end{remark}

\begin{table}
\footnotesize

\begin{tabular}{|l|C{2cm}|l|l|C{2cm}|l|l|C{2cm}|}
\cline{8-8}
\multicolumn{7}{c|}{} & $n = 50$ \\
\cline{5-5} \cline{7-8}
\multicolumn{4}{c|}{}&$n=30$ & &$ c < 3.88 $ & $50$-star \\
\cline{4-5} \cline{7-8}
\multicolumn{3}{c|}{}&$ c < 3.80 $ & $30$-star & & $c \in (3.88, 11.875)$&$25$-star \\
\cline{2-2}\cline{4-5}\cline{7-8}
\multicolumn{1}{c|}{}&$n=12$&&$c \in (3.80, 11)$ &$15$-star & & $c \in (11.875, 23.25)$&$17$-star \\
\cline{1-2}\cline{4-5}\cline{7-8}
$ c < 3.50 $ & $12$-star & & $c \in (11, 26)$ &$10$-star & &$c \in (23.25,30.(3))$&$13$-star \\
\cline{1-2}\cline{4-5}\cline{7-8}
$c \in (3.50, 9.50)$ &$6$-star & & $c \in (26,49)$ &$6$-star & &$c \in (30.(3), 85)$&$10$-star \\
\cline{1-2}\cline{4-5}\cline{7-8}
$c \in (9.50, 11.25)$ &$4$-star & & $c \in (49, 70.20)$ &$5$-star& &$c \in (85, 195)$ &$5$-star \\
\cline{1-2}\cline{4-5}\cline{7-8}
$c > 11.25$ &two disconnected components of equal size & & $c > 70.20$&two disconnected components of equal size& &$c > 195$ &two disconnected components of equal size \\
\cline{1-2}\cline{4-5}\cline{7-8}
\end{tabular}
\vspace*{0.5cm}

\caption{Optimal networks for $n \in \{12, 30, 50\}$.} \label{example50}
\end{table}

\subsection{Decentralized defense}
Now we turn attention to the variant of the model where defense decisions are decentralized. Our goal is to characterize the inefficiencies caused by decentralized protection decisions for general values of $\nbyz$ and $\ninf$. To this end, we need to compare equilibrium payoffs to the designer under centralized and decentralized defense. We start by establishing two results about the existence of equilibria in the decentralized defense game. 

Firstly, since the game is finite, we get equilibrium existence by Nash theorem. Notice that our use of the pessimistic aggregation of the incomplete information about types of nodes determines a game where the utilities of the nodes and the designer are defined by the corresponding pessimistic utilities. This game is finite and, by Nash theorem, it has a Nash equilibrium in mixed strategies. This leads to the following existence result.

\begin{proposition}
\label{pr:exist}
There exists an equilibrium of $\Gamma$.
\end{proposition}

\begin{proof}
It can be shown that a stronger statement holds. More precisely, one can prove that for any $n, c$ there exists an equilibrium $\equilibrium$ such that the strategies of the nodes do not depend on their types. Let us sketch the proof. We consider a modified model in which the nodes do not know their types (i.e., every node thinks that she is genuine, but some of them are byzantine). In this model, the (mixed) strategies of nodes are functions $\tilde{\delta}_{j} \colon \graphs(\nodes) \to \simplex(\{0,1\})$,\footnote{
If $X$ is a finite set, then by $\simplex(X)$ we denote the set of all probability measures on $X$.
}
and every node receives a pessimistic utility of a genuine node, as defined in \cref{eq:pess_payoff_node}. The strategies and payoffs to the adversary and the designer are as in the original model. Let $\overbar{x} \colon \graphs(\nodes) \times 2^{\nodes} \times \binom{\nodes}{\nbyz} \to \binom{\nodes}{\ninf}$ denote any optimal strategy of the adversary (i.e., a function that, given a defended network and the position of the byzantine nodes $\byznodes$, returns a subset of nodes that is optimal to infect in this situation). If we fix $\overbar{x}$, then the game turns into a two stage game (the designer makes his action first and then the nodes make their actions) with complete information. Therefore, this game has a subgame perfect equilibrium in mixed strategies. This equilibrium, together with $\overbar{x}$, forms an equilibrium $\equilibrium$ in the original model, because, in the original model, a byzantine node cannot improve her payoff by a unilateral deviation.
\end{proof}

Fix the parameters $\nbyz, \ninf$ and let $\mathcal{E}(n, c)$ denote the set of all equilibria of $\Vargamma$ with $n$ nodes and the cost of protection $c > 0$. Let $\globdespayoff(n, c)$ denote the best payoff the designer can obtain in the centralized defense game (as discussed in \cref{sec:centralized}).
The \emph{price of anarchy} is the fraction of this payoff over the minimal payoff to the designer that can be attained in equilibrium of $\Vargamma$ (for the given cost of protection $c$),
\begin{equation*}
\poa(n, c) = \frac{\hat{U}^{\des}_{\star}(n,c)}{\min_{\bm{e} \in \mathcal{E}(n,c)} \Ex\hat{U}^{\des}(\bm{e})} \, .
\end{equation*}

Although pure strategy equilibria may not exist for some networks, they always exist on generalized stars. Moreover, when these stars are large enough, by choosing such a star, the designer can ensure that all genuine core nodes are protected. 
This is enough to characterize the price of anarchy as $n$ goes to infinity (with a fixed cost $c$). The next proposition characterizes equilibria on generalized stars.

    \begin{proposition}\label{thm:ROEchar}
    Let $\equilibrium \in \Equilibria$ be any equilibrium of $\Vargamma$. Let $G = (V, E)$ be a generalized $k$-star. Denote $\card{V} = n$, $x = \left\lfloor \frac{n}{k} \right\rfloor - \ninf + 1$, and $y = n - \nbyz\left\lfloor \frac{n}{k} \right\rfloor$. Furthermore, suppose that $n \ge k \ge \nbyz+1$ and $x \ge 2$. If the cost value $c$ belongs to one of the intervals $(0, f(1))$, $(f(1), \frac{f(x)}{x})$, $(\frac{f(y)}{y}, +\infty)$, then the following statements about $\equilibrium$ restricted to $\Vargamma(G)$ hold:
     \begin{compactitem}
     \item all genuine nodes use pure strategies
      \item if $c < f(1)$, then all genuine nodes are protected
      \item if $f(1) < c < \frac{f(x)}{x}$, then 
            all genuine core nodes are protected and 
            all genuine periphery nodes are not protected
      \item if $\frac{f(y)}{y} < c$, then all genuine nodes are not protected.
    \end{compactitem}
  \end{proposition}
  
  The proof of \cref{thm:ROEchar} requires an auxiliary lemma.

  \begin{lemma}\label{attack_on_byz}
  Let $\equilibrium \in \Equilibria$ be any equilibrium of $\Vargamma$ and $\overbar{x} \colon \graphs(\nodes) \times 2^{\nodes} \times \binom{\nodes}{\nbyz} \to \simplex(\binom{\nodes}{\ninf})$ denote the (possibly mixed) strategy of the adversary in this equilibrium.
  Let $(G, \Delta)$ be a network such that $G$ is a generalized $k$-star. Furthermore, suppose that $\lfloor \frac{n}{k} \rfloor \ge 2$, $n \ge 3$, and that the set of byzantine nodes $\byznodes$ contains a core node. Then, $\overbar{x}(G, \Delta, \byznode)$ infects this node with probability one. 
  \end{lemma}
  
  \begin{proof}
  Since $\equilibrium$ is an equilibrium and the adversary has complete information about the network before making his decision, his strategy $\overbar{x}(G, \Delta, \byznode)$ is a probability distribution over the set of subsets of nodes that are optimal to attack. Let $b \in \byznodes$ denote any byzantine node that is also a core node. We will show that any optimal attack infects $b$.
  
To do so, fix any set of attacked nodes $I \in \binom{\nodes}{\ninf}$ and suppose that attacking $I$ does not infect $b$. Given the structure of generalized $k$-star, we see that $I$ consists of genuine protected nodes and periphery nodes that are connected to genuine protected core nodes. To finish the proof, fix any node $j \in I$ and observe it is strictly better for the adversary to attack the set $I \cup \{ b\} \setminus \{j\}$. Indeed, if $j$ is a genuine protected node, then attacking it does nothing, while attacking $b$ destroys at least one more node. Moreover, if $j$ is a periphery node connected to a genuine core protected node, then attacking $b$ not only destroys one node but also disconnects the network ($b$ is connected to at least one periphery node because $\lfloor \frac{n}{k} \rfloor -1 \ge 1$).
  \end{proof}

We are now ready to present the proof of \cref{thm:ROEchar}.
  
  \begin{proof}[Proof of \cref{thm:ROEchar}]
  Let $\overbar{x} \colon \graphs(\nodes) \times 2^{\nodes} \times \binom{\nodes}{\nbyz} \to \simplex(\binom{\nodes}{\ninf})$ denote the strategy of the adversary in $\equilibrium$ and let $\defense$ be any choice of protected nodes on $G$, $\defense \subset V$. Let $j \in V$ be a genuine node. 
   
   First, suppose that $j \notin \defense$. We will show that the pessimistic payoff of $j$ is equal to $0$. On the one hand, this payoff is nonnegative for every possible choice of the infected node. On the other hand, we can bound it from above by supposing that there exists a byzantine node $b \in \byznodes$ that is a core node and a neighbor of $j$. Then, \cref{attack_on_byz} shows that $\overbar{x}$ infects $b$, and the pessimistic payoff of $j$ is not greater than $0$. 
   
   Second, suppose that $j \in \defense$. Then, we have two possibilities. If $j$ is a periphery node, then the same argument as above shows that the pessimistic payoff of $j$ is equal to $f(1) - c$. If $j$ is a core node, then her payoff is bounded from below by $\frac{f(x)}{x} - c$ (where $x = \lfloor \frac{n}{k} \rfloor - \ninf + 1$) for every possible choice of the set of infected nodes. Moreover, by supposing that every byzantine node is a core node, we see that the pessimistic payoff of $j$ is bounded from above by $\frac{f(y)}{y} - c$ (where $y = n - \nbyz\lfloor \frac{n}{k} \rfloor$).
   
   Since the estimates presented above are valid for any choice of $\defense$, we get the desired characterization of equilibria.
  \end{proof}
  
Our main result estimates the price of anarchy using \cref{thm:ROEchar}.
\begin{theorem}
\label{th:poa}
Suppose that for all $t \ge 0$ the function $f$ satisfies 
\[ 
\lim_{n \to +\infty} f(n)/f(n - t) = 1 \, .
\]
Then, for any cost level $c > 0$ and any fixed parameters $\nbyz \ge 1$, $\ninf \ge 1$ we have 
\[
\lim_{n \to +\infty} \poa(n,c) = 1 \, .
\]
\end{theorem}

The proof requires an auxiliary lemma concerning the asymptotic behavior of the function $f$.
  
    \begin{lemma}\label{fast_growth}
  We have $\lim_{x \to +\infty} \frac{f(x)}{x} = +\infty$.
  \end{lemma}
   \begin{proof}
   Since $f$ is strictly convex, for any $0 < x < y < z$ we have (cf. \cite[Sect.~I.1.1]{lemarechal_convex_analysis})
  \begin{equation}
  \frac{f(y) - f(x)}{y - x} < \frac{f(z) - f(x)}{z - x} < \frac{f(z) - f(y)}{z - y} \, . \label{eq:slopes}
  \end{equation}
  As a result, the function $g_{t}(x) = (f(x + t) - f(t))/x$ is strictly increasing for all $t > 0$ (to see that, let $0 < x < y$ and use the left inequality from \cref{eq:slopes} on the tuple $(t, x + t, y + t)$). Since $f$ is convex and increasing, it is also continuous on $[0, +\infty)$ (cf. \cite[Sect.~I.3.1]{lemarechal_convex_analysis}). By fixing $x$ and taking $t \to 0$ we get that the function $x \to \frac{f(x)}{x}$ is nondecreasing. Suppose that $\lim_{x \to +\infty} \frac{f(x)}{x} = \eta < +\infty$. Then, by the assumption that $f(3x) \geq 2f(2x)$ for all $x \ge 1$, we have
  \[
  \eta = \lim_{x \to +\infty} \frac{f(x)}{x} = \lim_{x \to +\infty} \frac{f(3x)}{3x} \ge \lim_{x \to +\infty} \frac{2 f(2x)}{3/2 \cdot 2x} = 
  \frac{4}{3}\eta \, .
  \]
  Hence $\eta \le 0$ and $f(x) = 0$ for all $x \ge 0$, which contradicts the assumption that $f$ is strictly convex.
  \end{proof}
  
We also need the fact that $f$ is superadditive.
  \begin{lemma}\label{superadditive}
  For all $x, y > 0$ we have $f(x + y) > f(x) + f(y)$.
  \end{lemma}
  \begin{proof}
  From the strict convexity of $f$ we have $f(x) = f(\frac{x}{x+y}(x+y) + \frac{y}{x+y} \cdot 0) < \frac{x}{x+y}f(x+y)$. Analogously, $f(y) < \frac{y}{x+y}f(x+y)$. Hence $f(x+y) > f(x) + f(y)$. 
  \end{proof}

We now give the proof of \cref{th:poa}.

\begin{proof}[Proof of \cref{th:poa}]
The function $f$ is superadditive by \cref{superadditive}. As a result, the pessimistic payoff to the designer can be trivially bounded by $\globdespayoff(n,c) \le f(n)$. We now want to give a lower bound for the quantity $\min_{\bm{e} \in \mathcal{E}(n,c)} \Ex\hat{U}^{\des}(\bm{e})$. By \cref{fast_growth} we have $\lim_{x \to +\infty} \frac{f(x)}{x} = +\infty$. Let $N \ge 1 + \ninf$ be a natural number such that $\frac{f(x)}{x} > c$ for all $x \ge N - \ninf + 1$. For any $n \ge (\nbyz+1)(N + 1)$ we define $k = \lfloor \frac{n}{N + 1} \rfloor \ge \nbyz+1$. Observe that if we denote $x = \lfloor \frac{n}{k} \rfloor - \ninf + 1$, then we have $x \ge \frac{n}{k} - \ninf \ge N - \ninf + 1$.
 Hence, if the designer chooses a generalized $k$-star, then \cref{thm:ROEchar} shows that all genuine core nodes are protected in any equilibrium. In particular, we have
    $
  \min_{\bm{e} \in \mathcal{E}(n,c)} \Ex\hat{U}^{\des}(\bm{e}) \ge f(n - \nbyz\left\lceil \frac{n}{k} \right\rceil - \ninf + 1) - nc
  $. 
  Moreover, we can estimate
 \begin{equation}
  \begin{aligned}
  \left\lceil \frac{n}{k} \right\rceil \le \frac{n}{k} + 1 &\le \frac{n}{\frac{n}{N+1} - 1} + 1 \\ &\le \frac{n}{\frac{n}{N+1} - \frac{n}{2(N+1)}} + 1 = 2N + 3 \, .
  \end{aligned}
 \end{equation}
 Hence, using \cref{fast_growth}, we get
  \begin{align*} 
  \lim_{n \to +\infty} \PoA(n,c) &\le \lim_{n \to +\infty} \frac{f(n)}{f(n - \nbyz(2N + 3) - \ninf + 1) - nc}  \\ &= \lim_{n \to +\infty} \frac{1}{\frac{f(n - \nbyz(2N + 3) - \ninf + 1)}{f(n)} - \frac{nc}{f(n)}} = 1 \, . \qedhere
  \end{align*}
\end{proof}

\begin{remark}
Notice that the condition of \cref{th:poa} is verified for $f(x) = x^{a}$ with $a \geq 2$. 
Hence, in the case of such functions $f$, the price of anarchy is $1$, so the inefficiencies due to decentralization are fully mitigated by the network design.
This is true, in particular, for Metcalfe's law.
\end{remark}

\section{Extensions of the model}\label{sec:extension}

In the previous section, we have shown that the topology of generalized $k$-star mitigates the costs of decentralization in our model. Nevertheless, our approach can be used to show similar results in a number of modified models. For instance, one could consider a probabilistic model, in which $\nbyz$ byzantine nodes are randomly picked from the set of nodes $\nodes$ (and the distribution of this random variable is known to all players). Then, the designer and nodes optimize their expected utilities, not the pessimistic ones (where the expectation is taken over the possible positions of the byzantine nodes). In this case, we still can give a partial characterization of Nash equilibria on generalized $k$-stars. More precisely, one can show that if the assumptions of \cref{thm:ROEchar} are fulfilled and $f(1) < c < \frac{f(x)}{x}$, then all genuine core nodes are protected. This is exactly what we need in the proof of \cref{th:poa}. Therefore, the price of anarchy in the probabilistic model also converges to $1$ as the size of the network increases. 

\section{Conclusions}\label{sec:concl}
We studied a model of network defense and design in the presence of an intelligent adversary and byzantines nodes that cooperate with the adversary. We characterized optimal defended networks in the case where defense decisions are centralized, assuming that the number of byzantine nodes and the number of attacked nodes are equal to one. We have also shown that, in the case of sufficiently well-behaved functions $f$ (including $f$ in line with Metcalfe's law), careful network design allows to fully mitigate the inefficiencies due to decentralized protection decisions, despite the presence of the byzantine nodes. In terms of network design, we showed that a generalized star is a topology that can be used to achieve this goal. This topology creates incentives for protection by two means. Firstly, it is sufficiently redundant, so that the protected nodes are connected to several other protected nodes. This secures adequate network value even if some of these nodes are malicious. Secondly, it gives sufficient exposure to the nodes, encouraging the nodes that would benefit from protection to choose to protect through fear of being infected (either directly or indirectly). These results could be valuable, in particular, to policy-makers and regulators, showing that such regulations can have strong effect and providing hints for which network structures are better and why.

An interesting avenue for future research is to consider a setup where not only the identities but also the number of byzantine nodes are unknown. How would the optimal networks look like if the protection decisions are centralized? Can we still mitigate the inefficiencies caused by decentralization? Another interesting problem are the optimal networks under centralized protection when the number of byzantine nodes or the budget of the adversary are greater than~$1$. Based on our experiments, we suspect that the topology of these networks is very similar to the case considered here. Nevertheless, a formal result remains elusive.

\bibliographystyle{alpha}
\newcommand{\etalchar}[1]{$^{#1}$}

\newpage

\begin{appendix}
\section{Characterization of equilibria in the centralized defense model}\label{ap:centralized}
 In this section we prove the characterization of equilibira given in \cref{th:centr}. We start with some auxiliary lemmas.
  
    \begin{lemma}\label{slopes}
  For every $t > 0$, the function $\hat{g}_{t} \colon \RR_{> 0} \to \RR_{> 0}$ defined as $\hat{g}_{t}(x) = f(x + t) - f(x)$ is strictly increasing.
  \end{lemma}
  \begin{proof}
  Let $0 < x < y$. First use the left inequality from \cref{eq:slopes} on the tuple $(x, x + t, y + t)$ and then use the right inequality on the tuple $(x, y, y + t)$.
  \end{proof}
  
  \begin{corollary}\label{moving_half}
  For all $y \ge 1$ and $x \ge 2y$ we have $f(x + y) \ge f(x) + f(2y)$. For all $y \ge 1$ and $x \ge 2y + 2$ we have $f(x + y) \ge f(x) + f(2y + 1)$.
  
  \end{corollary}
  
  \begin{proof}
    Both claims follow from \cref{slopes} applied to $\hat{g}_{y}$ and \cref{eq:move_half}.   \end{proof}
  
\begin{figure*}[t]
\begin{center}
\begin{minipage}{0.5\textwidth}
  \centering
  \includegraphics[scale=0.5]{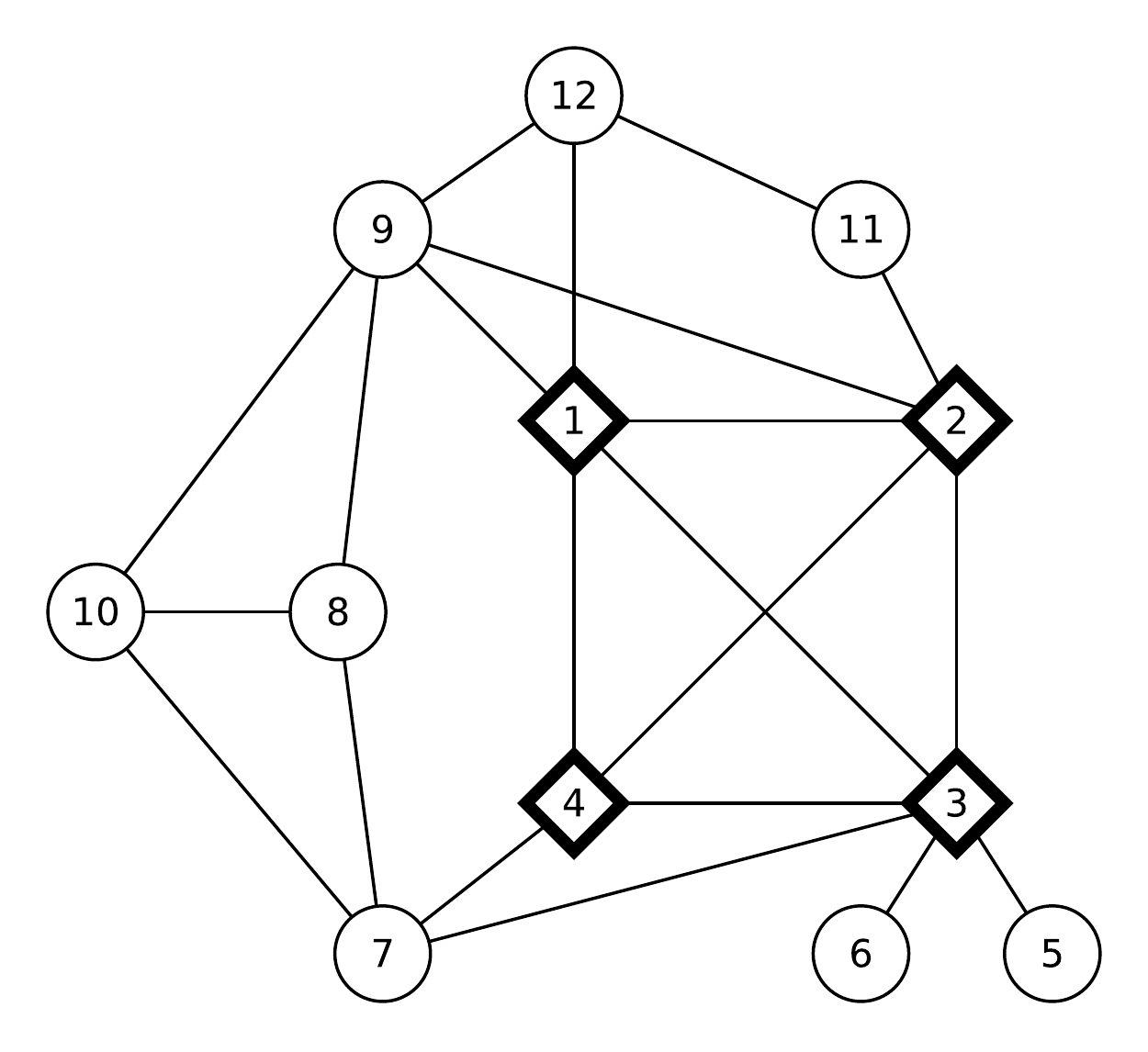}
\end{minipage}\hfill 
\begin{minipage}{0.5\textwidth}
  \centering
  \includegraphics[scale=0.5]{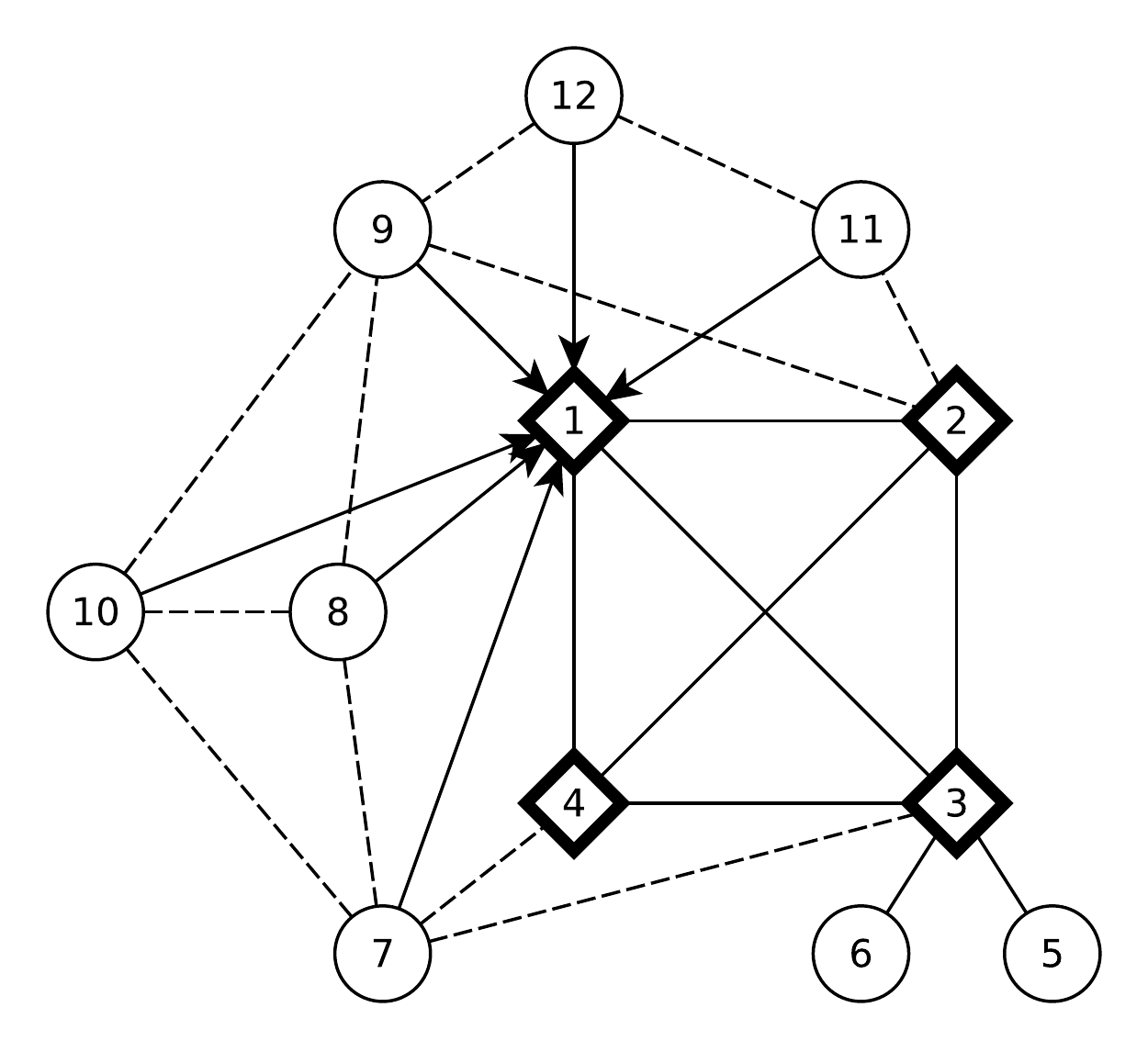}
\end{minipage}\hfill 
\end{center}
\caption{Transforming a network into a generalized star. Protected nodes are depicted in bold. In the left picture we have $V_{1} = \{1, 7,8,9,10,11,12\}$, $V_{3} = \{3,5,6,7,8,9,10,11,12\}$, and $V_{5} = \{5\}$. The node $7$ is reachable from $4$ in $\graph - V_{5}$. In the right picture the nodes $\{7,8,9,10,11,12\}$ have degree $1$ and are connected to node $1$. The node $7$ is still reachable from $4$ in $\tilde{\graph} - \tilde{V}_{5}$.}\label{fig:cloud_creation}
\end{figure*} 
  
  \begin{lemma}\label{lemma:star_structure}
  Suppose that $(\graph, \defense)$ is an equilibrium network and that $\card{\defense} = k \ge 2$. Then, there is a network $(\graph', \defense')$ such that $(\graph', \defense')$ is also an equilibrium network, $\card{\defense'} = k$, all nodes from $\defense'$ belong to the same connected component, this component is a generalized $k$-star, and $\defense'$ is the core of this star. Furthermore, the component of $\graph$ that contains $\defense$ is strictly larger than other components of $\graph$. 
  \end{lemma}
  \begin{proof}
  We will show how to transform $(\graph, \defense)$ into $(\graph', \defense')$ without diminishing the pessimistic payoff to the designer. First, if two nodes $i,j \in \defense$ are protected, then we add an edge between them. This does not decrease the designer's payoff, because there is only one byzantine node; hence, any attack infects at most one of the nodes $i,j$ and the residual network after the attack is not smaller that before the addition of the edge. Therefore, we can suppose that the subnetwork $\graph[\defense]$ is a clique. We focus on the connected component $\comp$ of $\graph$ that contains this clique. We will show that the remaining nodes of $\comp$ can be distributed in such a way that they form a periphery of a generalized $k$-star. 
  
  Let $G = (V, E)$. For any $i \in V$, let $V_{i} \subset V$ denote the set of nodes that get infected if $i$ is byzantine and gets infected. In other words, $V_{i}$ contains $i$ and all unprotected nodes $j \in V \setminus \defense$ such that there is a path from $i$ to $j$ that passes only through unprotected nodes. We refer to \cref{fig:cloud_creation} for an example. Observe that any optimal attack of the adversary that infects a node from $\graph$ infects in fact a set of nodes $V_{j}$. Indeed, if this attack infects the byzantine node $\byznode$, then the set of infected nodes is equal to $V_{\byznode}$. If, instead, this attack infects an unprotected genuine node $i$, then the set of infected nodes is equal to $V_{i}$. 
  
  We do the following operation. We fix $i \in \defense$, we take all unprotected nodes that belong to $V_{i}$, we delete all of their outgoing edges and, for every such node $j$, add the edge $ij$. An example of this operation is depicted in \cref{fig:cloud_creation}. We will show that this operation does not decrease the pessimistic payoff to the designer. Denote the new network by $\tilde{\graph} = (V, \tilde{E})$, and the corresponding sets by $\tilde{V}_{\ell}$ for $\ell \in V$. By the discussion in the preceding paragraph, it is enough to prove that for every $\ell \in V$, the connected components of the network $\graph - V_{\ell}$ do not get smaller after our operation. Suppose that $j_{0}j_{1} \in E$ is an edge in $\graph - V_{\ell}$ for some $\ell \in V$. We will prove that the node $j_{1}$ is still reachable from $j_{0}$ in the network $\tilde{\graph} - \tilde{V}_{\ell}$. First, we need to prove that $j_{0}, j_{1}$ do not belong $\tilde{V}_{\ell}$. Indeed, if $\ell = i$, then the claim is obvious because $\tilde{V}_{i} = V_{i}$. Otherwise, a path from $\ell$ to $j_{p}$ (for $p \in \{0,1\}$) that goes through unprotected nodes in $\tilde{\graph}$ cannot contain a node from $\tilde{V}_{i}$, because unprotected nodes in $\tilde{V}_{i}$ have degree $1$ and are connected to a protected node $i$. Thus, any such path does not contain a node from $V_{i}$, and hence it is also a path in $\graph$. Therefore $j_{0}, j_{1} \notin \tilde{V}_{\ell}$. We can now prove that $j_{1}$ is reachable from $j_{0}$ in $\tilde{\graph} - \tilde{V}_{\ell}$. If $j_{0}, j_{1} \notin V_{i}$, then $j_{0}j_{1}$ is an edge in $\tilde{E}$ and the claim is true. Otherwise, we have two possibilities. If both nodes $j_{0}, j_{1}$ belong to $V_{i}$, then $j_{0}ij_{1}$ is a path in $\tilde{\graph}$. If only one of them belongs to $V_{i}$, then the second one must belong to $\defense$, and hence $j_{0}ij_{1}$ is still a path in $\tilde{\graph}$ (because protected nodes form a clique in $\tilde{\graph}$). Moreover, the node $i$ does not belong to $V_{\ell}$ because $\ell \neq i$. Therefore, the path $j_{0}ij_{1}$ belongs to $\tilde{\graph} - \tilde{V}_{\ell}$. We can repeat this reasoning for every edge in $\graph - V_{\ell}$. As a consequence, if two nodes $j,j' \in V$ are connected by a path in $\graph - V_{\ell}$, then they are still connected by a path in $\tilde{\graph} -\tilde{V}_{\ell}$. Therefore, our operation does not decrease the pessimistic payoff of the designer.
  
  We can repeat the operation presented above for every protected node $i \in \defense$. As a result, we get a network $(\graph, \defense)$ such $\graph[\defense]$ is a clique and every unprotected node that belongs to the component $\comp$ containing this clique has degree $1$. It remains to prove that these nodes can be distributed evenly among the core protected nodes. Suppose that there are two protected nodes $i, j \in \defense$ such that $\abs{V_{i}} \ge \abs{V_{j}} + 2$ (where the sets $V_{\ell}$ are defined as previously). We take an unprotected node $\ell \in V_{i}$, delete the edge $i \ell$ and add the edge $j \ell$. This operation does not decrease the pessimistic payoff to the designer. Indeed, if the adversary infects a node in a component different than $\comp$, then the payoff to the designer does not change. Otherwise, the pessimistic utility to the designer is achieved when the adversary infects a byzantine node $i^{*} \in \defense$ such that the set $V_{i^{*}}$ has maximal cardinality. Hence, this payoff does not decrease after our operation.
  
 Finally, if the component of $\graph$ that contains $\defense$ is smaller than or equal to a component that does not contain any protected node, then it is more profitable to the adversary to infect this unprotected component. Hence, the designer can strictly improve his payoff by not using any protection at all, $\defense = \varnothing$, which gives a contradiction with our assumptions.
    \end{proof}
    
\begin{figure*}[t]
\begin{center}
\begin{minipage}{0.5\textwidth}
  \centering
  \includegraphics[scale=0.5]{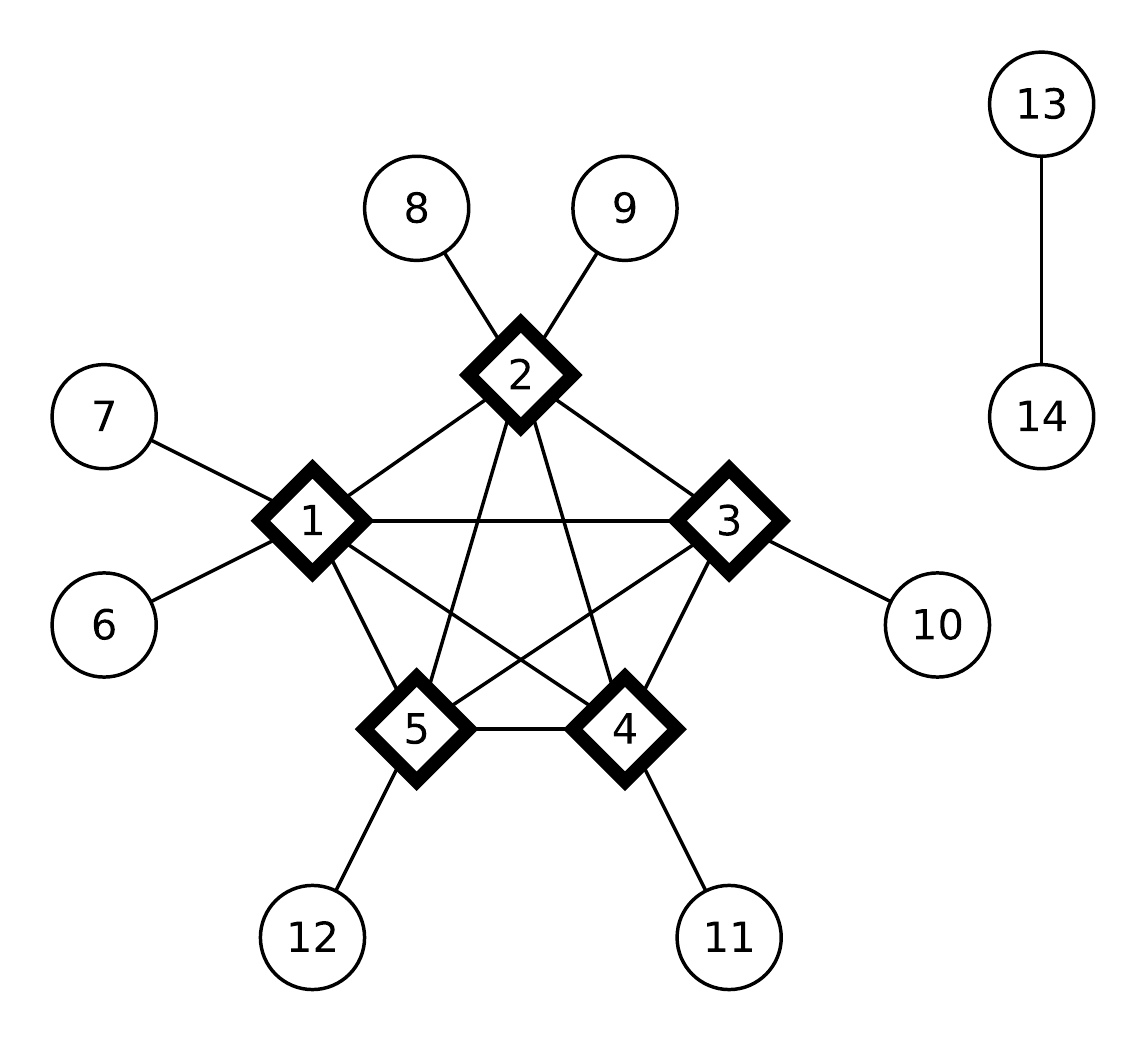}
\end{minipage}\hfill 
\begin{minipage}{0.5\textwidth}
  \centering
  \includegraphics[scale=0.5]{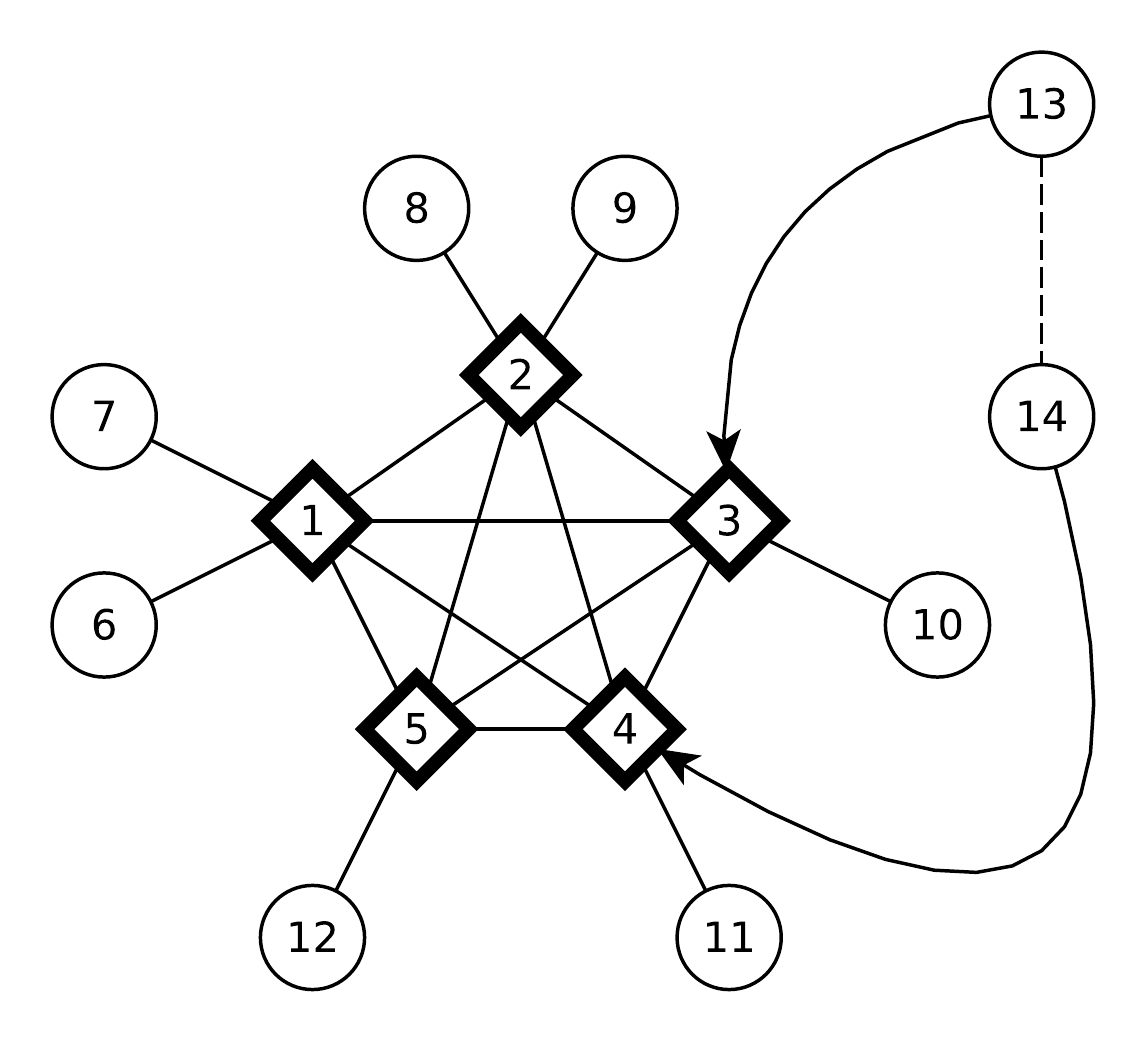}
\end{minipage}\hfill 
\end{center}
\caption{Spreading the nodes from an unprotected component to a core-protected generalized star.}\label{fig:cloud_spreading}
\end{figure*} 
  
  \begin{proof}[Proof of \cref{th:centr}]
    Let $(\graph, \defense)$ be an equilibrium network. Let $\defcomp, \comp_{2}, \dots, \comp_{m}$ denote its connected components, $s_{i} = \abs{\comp_{i}}$ for all $i \ge 1$, and assume that $s_{1} \ge s_{2} \ge \dots \ge s_{m}$. We will show how to transform $(\graph, \defense)$ into an equilibrium network with the topology described in the claim. Let $k = \abs{\defense}$ and observe that $k \neq 1$. Indeed, if $\abs{\defense} = 1$, then the pessimistic payoff to the designer is strictly larger if he stops protecting the one protected node (because, in the worst case scenario, this node is byzantine). First, suppose that $k \ge 3$. By \cref{lemma:star_structure}, we can assume that $\defcomp$ is a generalized $k$-star with protected core and that $s_{1} > s_{2}$. We want to find an equilibrium in which $s_{2} \in \{0,1\}$. If $s_{2} \ge 2$, then we consider the following transformation of the network $(\graph, \defense)$: we take the unprotected component $\comp_{2}$ and move all of its nodes to $\defcomp$, spreading them in a regular fashion as depicted in \cref{fig:cloud_spreading}. More formally, we consider a network $(\trgraph, \defense)$ such that $\trgraph$ consists of a connected component $\trdefcomp$ having a topology of a generalized $k$-star of size $s_{1} + s_{2}$ with protected core and unprotected components $\comp_{3}, \comp_{4}, \dots, \comp_{m}$. We will show that the network $(\trgraph, \defense)$ gives a payoff to the designer that is no smaller than the payoff obtained from choosing $(\graph, \defense)$. First, observe that if the adversary infects the component $\comp_{2}$ in $(\graph, \defense)$, then the pessimistic payoff to the designer is equal to $\Payoff_{1} = f(s_{1}) + f(s_{3}) + f(s_{4}) +  \dots + f(s_{m})$. Moreover, if the adversary infects the component $\defcomp$ in $(\graph, \defense)$, then the pessimistic payoff to the designer is equal to
    \begin{equation*}
    \Payoff_{2} = f\left(s_{1} - \left\lceil \frac{s_{1}}{k}\right\rceil \right) + f(s_{2}) + f(s_{3}) + \dots + f(s_{m}) \, .
    \end{equation*}
    Hence, the pessimistic payoff to the designer from playing $(\graph, \defense)$ is equal to $\min\{P_{1}, P_{2}\}$. On the other hand, his payoff from playing $(\trgraph, \defense)$ is equal to
     $\min \{T_{1}, T_{2}\}$, where $T_{1} =f(s_{1} + s_{2}) + f(s_{4}) + \dots + f(s_{m})$
     and 
     \begin{equation*}
     T_{2} = f\left(s_{1} + s_{2} - \left\lceil \frac{s_{1} + s_{2}}{k}\right\rceil \right) + f(s_{3}) + \dots + f(s_{m}) \, . 
     \end{equation*}
    Therefore, it is enough to show that $\min\{T_{1}, T_{2}\} \ge \min\{P_{1}, P_{2}\}$. By \cref{superadditive} we get $f(s_{1} + s_{2}) \ge f(s_{1}) + f(s_{2}) \ge f(s_{1}) + f(s_{3})$.
     This shows that $T_{1} \ge P_{1}$. To prove that $T_{2} \ge \min\{P_{1}, P_{2}\}$ we consider multiple cases, depending on the relative sizes of $\defcomp$ and $\comp_{2}$.
    
     \begin{asparaenum}[{Case} I:]
     \item Suppose that $s_{2} \ge 2\left\lceil\frac{s_{1}}{k}\right\rceil$. We then have $T_{2} \ge P_{1}$ by the inequality
      \begin{equation}\label{eq:some_technical}
      \begin{aligned}
      f\Bigl(s_{1} + s_{2} &- \left\lceil \frac{s_{1} + s_{2}}{k}\right\rceil \Bigr) \ge 
       f\left(s_{1} + s_{2} - \left\lceil \frac{s_{1}}{k}\right\rceil - \left\lceil \frac{s_{2}}{k}\right\rceil \right)\\ &\ge 
      f\left(s_{1} - \left\lceil \frac{s_{1}}{k}\right\rceil + \frac{s_{2}}{2} \right) \ge f(s_{1}) \, .
    \end{aligned}
    \end{equation}
    \item Suppose that $2 \le s_{2} \le s_{1} - \left\lceil \frac{s_{1}}{k}\right\rceil$. In this case, by \cref{moving_half},
    \begin{equation}\label{eq:some_technical2}
    f\left(s_{1} - \left\lceil \frac{s_{1}}{k}\right\rceil + \frac{s_{2}}{2} \right) \ge f\left(s_{1} - \left\lceil\frac{s_{1}}{k}\right\rceil\right) + f(s_{2}) \, .
    \end{equation}
    Thus, we have $T_{2} \ge P_{2}$ by combining \cref{eq:some_technical2} and the first two inequalities of \cref{eq:some_technical}.
     \item Suppose that $s_{1} - \left\lceil \frac{s_{1}}{k}\right\rceil < s_{2}  < 2\left\lceil\frac{s_{1}}{k}\right\rceil$ and $s_{2} \ge 2$. Let $s_{1} = kl + r$, where $0 \le r < k$ and $l \ge 1$. If $r = 0$, then we have $kl - l < 2l$, which is impossible for $k \ge 3$. Hence $r \ge 1$ and we have $kl + r - l - 1 < s_{2} < 2l + 2$. Note that the open interval $(kl + r - l - 1, 2l + 2)$ contains an integer number if and only if $(2l + 2) - (kl + r - l - 1) \ge 2 \iff 3l + 1 \ge kl + r$. This condition is satisfied only for $k = 3$ and $r = 1$. Hence, we have $s_{1} = 3l + 1$ and $s_{2} = 2l + 1$ for some $l \ge 1$. We want to prove that $T_{2} \ge P_{1}$ or, equivalently,
     \[
      f\left(5l + 2 - \left\lceil\frac{5l + 2}{3}\right\rceil \right) \ge  f(3l + 1) \, .
     \]
     If $l = 1$, then this inequality takes form $f(4) \ge f(4)$. If $l \ge 2$, then we have $5l + 2 \le 6l$ and hence $f\left(5l + 2 - \left\lceil\frac{5l + 2}{3}\right\rceil \right) \ge f(3l + 2) \ge  f(3l + 1)$.
     \end{asparaenum}
     Therefore, there is an equilibrium network such that $s_{2} \in \{0, 1\}$. If $s_{2} = 1$ and $m \ge 3$, then \cref{superadditive,moving_half} give $f(s_{1} - 1) + 2f(1) < f(s_{1} - 1) + f(2) \le f(s_{1})$. Therefore, it is more profitable for the adversary to infect one node from $\defcomp$ than to infect a component composed of two nodes. Thus, it would be strictly profitable to the designer to merge $\comp_{2}$ and $\comp_{3}$, which gives a contradiction. Hence, we have $s_{2} = 0$ or $s_{2} = 1$ and $m = 2$. It is easy to see that the first case is more profitable to the designer if $n \bmod k \neq 1$ while the second case is more profitable if $n \bmod k = 1$.
     
     The proofs for the cases $k = 0$ and $k = 2$ are less involved than the one above, so we just sketch them. For $k = 0$, the pessimistic payoff to the designer is equal to $\Payoff = f(s_{2}) + \ldots + f(s_{m})$.
     We do the following transformations on the network: if $s_{i} = 2l$ for some $i \ge 3$ and $l \ge 1$, then we spread half of $\comp_{i}$ into $\defcomp$ and the other half into $\comp_{2}$. By \cref{moving_half} we have $f(s_{2} + l) \ge f(s_{2}) + f(s_{i})$, and hence this change is profitable to the designer. If $s_{i}$ is odd for all $i \ge 3$ and we have $m \ge 4$, then we take all the nodes belonging to the union of $\comp_{3}$ and $\comp_{4}$ and spread half of them into $\defcomp$ and the other half into $\comp_{2}$. This improves the designer's payoff by the inequality $f((s_{2} + \frac{1}{2}s_{3}) + \frac{1}{2}s_{4}) \ge f(s_{2} + \frac{1}{2}s_{3}) + f(s_{4}) \ge f(s_{2}) + f(s_{3}) + f(s_{4})$. Finally, if $m = 3$ and $s_{3} = 2l +1$ is odd, greater than $1$, and strictly smaller than $s_{2}$, then we spread $l$ nodes from $\comp_{3}$ to $\defcomp$ and $l + 1$ nodes to $\comp_{2}$. By \cref{moving_half} we have $\min\{f(s_{1} + l), f(s_{2} + l + 1)\} \ge f(s_{2} + l) \ge f(s_{2}) + f(s_{3})$ and this change is profitable to the designer.
     
     For $k = 2$, we can suppose (as in the case $k = 3$), that $\defcomp$ is a generalized $2$-star with protected core and that $s_{1} > s_{2}$. The pessimistic payoff to the designer is equal to $\min\{\Payoff_{1}, \Payoff_{2}\}$, where
     \begin{align*}
     \Payoff_{1} = f(s_{1}) + f(s_{3}) + f(s_{4}) +  \dots + f(s_{m}) \, , \\
     \Payoff_{2} = f\left(s_{1} - \left\lceil \frac{s_{1}}{2}\right\rceil \right) + f(s_{2}) + f(s_{3}) + \dots + f(s_{m}) \, .
     \end{align*}
     We do the following transformation on the network: if $s_{3}  = 2l$, then we spread half of its nodes to $\comp_{2}$ and the other half to $\defcomp$ (so that $\defcomp$ becomes a generalized $2$-star with $s_{1} + l$ nodes). By \cref{moving_half} we have $f(s_{1} + l) \ge f(s_{1}) + f(s_{3})$ and $f(s_{2} + l) \ge f(s_{2}) + f(s_{3})$. Moreover, we have
     \begin{equation}
    \begin{aligned} \label{13}
    f\left(s_{1} + l - \left\lceil \frac{s_{1} + l}{2}\right\rceil \right) &\ge 
     f\left(s_{1} - \left\lceil \frac{s_{1}}{2}\right\rceil + l - \left\lceil \frac{l}{2}\right\rceil \right) \\ &\ge 
    f\left(s_{1} - \left\lceil \frac{s_{1}}{2}\right\rceil \right) \, .
    \end{aligned}
    \end{equation}
    
  Therefore, this change is profitable to the designer. If $s_{3} = 2l + 1$ is odd and greater that $1$, then we do the following transformation: we spread $l$ nodes to $\defcomp$ (so that $\defcomp$ becomes a generalized $2$-star with $s_{1} + l$ nodes) and $l + 1$ nodes to $\comp_{2}$. \Cref{13} still holds. Moreover, since $s_{1} > s_{2} \ge s_{3}$, \cref{moving_half} shows that $f(s_{2} + l + 1) \ge f(s_{2} + 1) + f(s_{3}) \ge f(s_{2}) + f(s_{3})$ and $f(s_{1} + l) \ge f(s_{1}) + f(s_{3})$. As before, this change is profitable to the designer. Finally, if $s_{3} = 1$ and $m \ge 4$, then we have two cases. If $s_{1} \ge 3$ then, by the same reasoning as in the case $k = 3$, merging $\comp_{3}$ and $\comp_{4}$ is profitable to the designer. Otherwise, we have $s_{1} = 2, s_{2} = 1$ and $s_{i} = 1$ for all $i \ge 3$. In this case, we have $f(s_{1}) = f(2) > 2f(1) = f(s_{1} - 1) + f(1)$. Therefore, the optimal attack of the adversary attacks the protected node. It is thus profitable to the designer to split the nodes forming $\defcomp$ and do not use the protection.
  
  To finish the proof, we observe that the quantities $\despayoff_{k}(n)$ correspond to the pessimistic payoffs of the designer achieved from choosing an equilibrium network with $k$ protected nodes and the topology described in the claim.
  \end{proof}
\end{appendix}

\end{document}